\documentclass[12pt,draft]{article}

\usepackage[utf8]{inputenc}
\usepackage[english]{babel}

\usepackage[square,numbers]{natbib}
\usepackage{url}
\usepackage{tabulary}

\usepackage{amsmath}
\usepackage{amsthm}
\usepackage{amssymb}

\title{From Graphs to Keyed Quantum Hash Functions}
\author{M.~Ziatdinov}
\date{May 28, 2016}

\newcommand{\ket}[1]{\left| #1 \right\rangle}
\newcommand{\bra}[1]{\left\langle #1 \right|}
\newcommand{\braket}[2]{\langle #1 | #2 \rangle}

\newcommand{\Hilb}{\mathcal{H}^2}
\newcommand{\Cplx}{\mathbf{C}}
\newcommand{\Xzo}{\{0,1\}}
\newcommand{\HH}[1]{(\Hilb)^{\otimes #1}}
\newcommand{\Hom}[1]{\mathrm{Hom}(#1,#1)}
\newcommand{\E}{\mathbf{E}}
\newcommand{\EE}{\Psi}
\newcommand{\ZZ}{\mathbb{Z}}
\newcommand{\CC}{\mathbb{C}}
\newcommand{\Supp}{\mathrm{Supp}}
\newcommand{\approxeps}{\overset{\epsilon}{\approx}}
\newcommand{\Ext}{\mathrm{Ext}}
\newcommand{\key}{\mathsf{key}}
\newcommand{\ExpF}{\EE_{\Gamma,t}}
\newcommand{\ExtF}{\EE_{\Ext,t,S}}
\newcommand{\ExtKF}{\EE_{\Ext}}

\DeclareMathOperator{\Tr}{Tr}

\newtheorem{defn}{Definition}
\newtheorem{thm}{Theorem}

\newtheorem{cor}{Corollary}
\newtheorem{question}{Problem}

\begin{document}

\maketitle{}

\begin{abstract}
  We present two new constructions of quantum hash functions: the first based on expander graphs and the second based on extractor functions and estimate the amount of randomness that is needed to construct them. We also propose a keyed quantum hash function based on extractor function that can be used in quantum message authentication codes and assess its security in a limited attacker model.  
\end{abstract}

\section{Introduction}\label{sec:introduction}

Quantum hash functions are similar to classical (cryptographic) hash functions and their security is guaranteed by physical laws. However, their construction and applications are not fully understood.

Quantum hash functions were first implicitly introduced in \citet{Buhrman2001} as quantum fingerprinting. Then \citet{Gavinsky2010} noticed that quantum fingerprinting can be used as cryptoprimitive. However, binary quantum hash function are not very suitable if we need group operations (and group is not $\ZZ_{2^k}$. For example, several classical hash functions were proposed that use groups, e.g. by \citet{Charles2009} and by \citet{Tillich1994}. 

\citet{Ablayev2015} gave a definition and construction of non-binary quantum hash functions. \citet{Ziatdinov2016} showed how to generalize quantum hashing to arbitrary finite groups. Recently, \citet{Vasiliev2016} showed how quantum hash functions are connected with $\epsilon$-biased sets.

Quantum hash functions map a classical message into a Hilbert space. Such space should be as small as possible, so eavesdropper can't read a lot of information about classical message (this is guaranteed by physical laws as Holevo-Nayak's theorem states). But images of different messages should be as far apart as possible, so recipient can check that hash differ or not with high probability. We measure this distance using an absolute value of scalar product of hashes of different messages.

Informally speaking, to define a quantum hash function we need some random data. Then our input is mixed with this random data. Quantum parallelism allows us to do it in different subspaces simultaneously, so resulting hash is small.
For example, random subsets suffice (for $\ZZ_m$) \citep{Ablayev2008}, random codes suffice (for $\ZZ_2^n$) \citep{Buhrman2001}, random automorphisms suffice (for any finite group) \citep{Ziatdinov2016}. \citet{Vasiliev2015} used some heuristics to find best subsets of $\ZZ_m$. 

However, typically the amount of randomness that is needed to construct such quantum hash functions is large (about $O(\log^2 |G|)$). We reduce amount of randomness needed to define quantum hash function to $O(\log |G| \log \log |G|)$ in expander-based quantum hash function.

Extractor-based quantum hash function allows us to introduce a notion of keyed quantum hash function. It can be used, for example, in quantum message authentication codes. Unlike \citep{Barnum2001} and \citep{Barnum2002} we use classical keys and authenticate classical messages. Unlike \citep{Curty2001} we authenticate whole messages, not single bits. However, our security analysis has only limited attacker.

It is known that walk on expander graph gives results very similar to random sampling. We show that walks on expander graphs give a quantum hash functions in section \ref{sec:expander-qhf}. Structure of these quantum hash functions is somewhat different from previous versions.

Extractor is a generalization of expander graph. 
In the section \ref{sec:keyed-qhf} we propose a keyed quantum hash function based on extractors and assess its security against limited attacker.


\paragraph{Acknowledgements.} I thank Farid Ablayev, Alexander Vasiliev and Marco Carmosino for helpful discussions.
A part of this research was done while attending a Special Semester Program on Computational and Proof Complexity (April-June 2016) organized by Chebyshev Laboratory of St.Petersburg State University in cooperation with Skolkovo Institute of Science and Technology and Steklov Institute of Mathematics at St.Petersburg. 
Partially supported by Russian Foundation for Basic Research, Grants 14-07-00557, 15-37-21160.
The work is performed according to the Russian Government Program of Competitive Growth of Kazan Federal University.

\section{Definitions}

Let us recall some basic definitions.

\subsection{Statistics}

We use a standard definition of the statistical distance.

\begin{defn}[Statistical distance, cited by \citet{Shaltiel2011}]
  We say that two distributions $F$ and $G$ are $\epsilon$-close, if for every event $A$, $|\Pr[F \in A] - \Pr[G \in A]| \le \epsilon$.

  The support of a distribution $X$ is $\Supp(X) = \{ x : \Pr[X = x] > 0 \}$.

  The uniform distribution over $\Xzo^m$ is denoted by $U_m$ and we say that $X$ is $\epsilon$-close to uniform if it is $\epsilon$-close to $U_m$.

  We denote that distribution $F$ is $\epsilon$-close to distribution $G$ by $F \approxeps G$.
\end{defn}

We also use a standard definition of the min-entropy.

\begin{defn}[Min-entropy, cited by \citet{Shaltiel2011}]
  Let $X$ be a distribution. The min-entropy of $X$ is $H_\infty(X) = \min_{x \in \Supp(X)} \log \frac 1 {\Pr[X=x]}$.
\end{defn}

\subsection{Quantum model of computation}

We use the following model of computation.

Recall that a qubit $\ket{\Psi}$ is a superposition of basis states $\ket{0}$ and $\ket{1}$, i.e. $\ket{\Psi} = \alpha\ket{0} + \beta\ket{1}$, where $\alpha, \beta \in \Cplx$ and $|\alpha|^2 + |\beta|^2 = 1$. So, qubit $\ket{\Psi} \in \Hilb$, where $\Hilb$ is a two-dimensional Hilbert complex space.

Let $s \ge 1$. We denote $2^s$-dimensional Hilbert complex space by $\HH{s}$:
\[
\HH{s} = \Hilb \otimes \Hilb \otimes \ldots \otimes \Hilb = \mathcal{H}^{2^s}
\]

We denote a state $\ket{a_1}\ket{a_2}\ldots\ket{a_n}$, each $a_i \in \Xzo$, by $\ket{i}$, where $i$ is $\overline{a_1a_2\ldots a_n}$ in binary. For example, we denote $\ket{1}\ket{1}\ket{0}$ by $\ket{6}$. Usually it is clear, which space this state belongs to.

Computation is done by multiplying a state by a unitary matrix: $\ket{\Psi_1} = U \ket{\Psi_0}$, where $U$ is a unitary matrix: $U^\dagger U = I$, $U^\dagger$ is the conjugate matrix and $I$ is the identity matrix.

The density matrix of a mixed state $\{p_i, \ket{\psi_i}\}$ is a matrix $\rho = \sum_i p_i \ket{\psi_i}\bra{\psi_i}$. A density matrix belongs to $\Hom{\HH{s}}$, the set of linear transformations from $\HH{s}$ to $\HH{s}$.

At the end of computation state is measured by POVM (Positive Operator Valued Measure). A POVM on a $\HH{s}$ is a collection $\{E_i\}$ of positive semi-definite operators $E_i : \Hom{\HH{m}} \to \Hom{\HH{m}}$ that sums up to the identity transformation, i.e. $E_i \succeq 0$ and $\sum_i E_i = I$. Applying a POVM $\{E_i\}$ on a density matrix $\rho$ results in answer $i$ with probability $\Tr(E_i \rho)$.

\subsection{Character theory}






\begin{defn}[Character of the group]
  Let $G$ be a group with unity $e$ and operation $\circ$.

  The character $\chi: G \to \CC$ of the group $G$ is a homomorphism of $G$ to $\CC$: for any $g, g' \in G$ it holds that $\chi(g \circ g') = \chi(g) \chi(g')$.
\end{defn}

\subsection{Graphs}


\begin{defn}[Expander graph, cited by \citet{Hoory2006}]
  Let the graph $\Gamma = (V,E)$ with set of vertices $V$ and set of edges $E$ be fixed. Self-loops and multiple edges are allowed.

  Graph $\Gamma$ is the $d$-regular graph if all vertices have the same degree $d$; i.e. each vertex is incident to exactly $d$ edges.

  Adjacency matrix of the graph $A = A(\Gamma)$ is an $n \times n$ matrix whose $(u,v)$ entry is the number of edges between vertex $u$ and vertex $v$.

  Let $\lambda_1 \ge \lambda_2 \ge \ldots \ge \lambda_n$ be eigenvalues of matrix $A = A(\Gamma)$, i.e. for some $v_i$ it holds that $A v_i = \lambda_i v_i$. We refer to the eigenvalues of $A(\Gamma)$ as the spectrum of the graph $\Gamma$.

  Given a $d$-regular graph $\Gamma$ with $n$ vertices and spectrum $\lambda_1 \ge \lambda_2 \ge \ldots \ge \lambda_n$ we denote $\lambda(\Gamma) = \max\{ |\lambda_2|, |\lambda_n| \}$.

  We call the graph $\Gamma$ a $(d,\lambda)$-expander graph if $\Gamma$ is $d$-regular and has $\lambda(\Gamma) = \lambda$.
\end{defn}

Every expander graph can be converted to a bipartite expander graph. One can just take two copies of vertex sets and change original edges to go from one copy to another. Generalization of these bipartite expander graphs is extractor graphs. The extractor graph is a bipartite graph where size of components can be different. An extractor can also be defined in terms of function that maps pair of first component vertex and edge to second component vertex.

\begin{defn}[(Seeded) extractor, cited by \citet{Shaltiel2011}]
  A function $E: \Xzo^n \times \Xzo^d \to \Xzo^m$ is a $(k,\epsilon)$-extractor if for every distribution $X$ over $\Xzo^n$ with $H_\infty(X) \ge k$, $E(X,Y)$ is $\epsilon$-close to uniform (where $Y$ is distributed like $U_d$ and is independent of $X$).
\end{defn}

Sometimes we use extractor functions that map one (arbitrary) set to other: $E: G \times \Xzo^d \to H$. These functions can be thought of as bipartite graphs with vertices $(G,H)$. In this case we denote uniform distribution on $H$ by $U_H$. 

We also use extractors against quantum storage. Informally, their output is $\epsilon$-close to uniform and no quantum circuit operating on $b$ qubits can distinguish output from uniform.

\begin{defn}[Extractor against quantum storage, cited by \citet{Ta-Shma2009}]
  An $(n,b)$ quantum encoding is a collection $\{\rho(x)\}_{x \in \Xzo^n}$ of density matrices $\rho(x) \in \HH{b}$.

  A boolean test $T$ $\epsilon$-distinguishes a distribution $D_1$ from a distribution $D_2$ if $|\Pr_{x_1 \in D_1}[T(x_1) = 1] - \Pr_{x_2 \in D_2}[T(x_2) = 1]| \ge \epsilon$.

  We say $D_1$ is $\epsilon$-indistinguishable from $D_2$ if no boolean POVM can $\epsilon$-distinguish $D_1$ from $D_2$.

  A function $X: \Xzo^n \times \Xzo^d \to \Xzo^m$ is a $(k,b,\epsilon)$ strong extractor against quantum storage, if for any distribution $X \subseteq \Xzo^n$ with $H_\infty(X) \ge k$ and every $(n,b)$ quantum encoding $\{\rho(x)\}$, $U_t \circ E(X,U_t) \circ \rho(X)$ is $\epsilon$-indistinguishable from $U_{t+m} \circ \rho(X)$.
\end{defn}

\section{Quantum hash functions}\label{sec:qhf}

Informally, quantum hash function is a function that maps {\em large} classical input to a {\em small} quantum (hash) state such that two requirements are satisfied: (1) it is hard to restore input given the hash state and (2) it is easy to check with high probability that inputs for two quantum hash states are equal or different.

It is easy to meet the first requirement for a constant hash size. One can simply take a qubit $\ket{\Psi(w)} = \alpha(w)\ket{0} + \beta\ket{1}$ and encode the input in a fractional part of $\alpha$. But then the second requirement is not satisfied.

It is easy to meet the second requirement for a hash size that is logarithmic in input size. One can simply map the input to the corresponding base state: $\ket{\Psi(i)} = \ket{i}$. However, then the first requirement is not satisfied.

Let us give the formal definition.

\begin{defn}[Quantum hash function, cited by \citet{Ablayev2014}]
  For $\delta \in (0, 1/2)$ we call a function $\psi : X \to \HH{s}$ a $\delta$-resistant function if for any pair $w,w'$ of different elements of $X$ their images are almost orthogonal:
  \begin{equation}\label{eq:qhf-resistance}
  |\braket{\psi(w)}{\psi(w')}| \le \delta.
  \end{equation}

  We call a map $\psi: X \to \HH{s}$ an $\delta$-resistant $(K;s)$ quantum hash function if $\psi$ is a $\delta$-resistant function, and $\log|X| = K$.
\end{defn}

Quantum hash function maps inputs of length $K$ to (quantum) outputs of length $s$. If $K \gg s$ any attacker can't get a lot of information by Holevo-Nayak theorem \cite{Nayak1999}.

The equality of two hashes can be checked using, for example, well-known SWAP-test \citep{Gottesman2001}.

All our hash functions have the following form:
\begin{equation}\label{eq:qhf-structure}
  \ket{\psi(g)} = \sum_{i=1}^t \chi(k_i(g)) \ket{i},
\end{equation}
where $g$ is an element of some group $G$, $\{k_i, i=1,\ldots,t\}$, $k_i: G \to H$ is a set of mappings from group $G$ with operation $\circ$ to group $H$ with operation $\bullet$ and $\chi: H \to \CC$ is a character of the group $H$.

For example, the group $G$ can be thought of as $Z_{2^n}$ with group operation $+$, then elements of $G$ can be encoded as binary strings $\Xzo^n$ of length $n$

\subsection{Why groups?}\label{sec:why-group}

We use groups in quantum hash functions of form (\ref{eq:qhf-structure}), not just arbitrary sets, because groups have nice structure. We can combine elements of group and we can inverse them. 

Several classical cryptoprimitives were proposed that use groups, e.g. by \citet{Charles2009} and by \citet{Tillich1994}.

\section{Expanders for Quantum Hashing}\label{sec:expander-qhf}

As noted in Section \ref{sec:introduction}, randomly chosen parameters with high probability lead to a quantum hash function. We replace this process with random walk on expander graph that is known to be close to uniform sampling.

In this section we fix a group $G$ with group operation $\odot$ and unity $e$. 

Let $\Gamma = (V,E)$ be an extractor - i.e. $d$-regular graph with spectral gap $\lambda$. We label vertices $V$ of graph $\Gamma$ with messages (i.e. elements of group $G$).

Let us randomly choose one vertex and perform a random walk of length $t$ starting from it. Denote vertices that occured in this walk by $s_j$. Parameter $t$ depend on security parameter $\epsilon$ of quantum hash function and we derive its value in theorem \ref{thm:t-expander}.

It is easy to note that such construction requires only $t d + \log |G|$ bits of randomness.

Let us define the expander quantum hash function.

\begin{defn}
  The expander quantum hash function $\ExpF(g)$ maps elements of $G$ to unitary transformations in $m$-dimensional Hilbert space $\HH{m}$:
  \[
  \ExpF(g) = \sum_{k=1}^t \chi(g \odot s_k) \ket{k}.
  \]
\end{defn}

If we choose $\Gamma$ and $t$ appropriately, $\ExpF$ is a quantum hash function.

\begin{thm}\label{thm:t-expander}
  For any $\delta \in (0; \frac 1 2)$ the function $\ExpF$ is a $\delta$-resistant $(\log |G|;\log t)$ quantum hash function if $t > O(\frac{\log |G|}{\delta})$.
\end{thm}

\begin{proof}
  Let us fix some $t$. 
  
  \[
  \braket{\EE^\dagger(g)}{\ExpF(g')} = \sum_{k=1}^t \braket{\chi^*(g \odot s_k)} {\chi(g' \odot s_k)} = |\sum_{k=1}^t \chi(s_k^{-1} \odot g^{-1} \odot g' \odot s_k)|.
  \]

  Denoting $g'' = g^{-1} \odot g'$, we get
  \[
  \braket{\EE^\dagger(g)}{\ExpF(g')} = |\sum_{k=1}^t \chi(s_k^{-1} \odot g'' \odot s_k)|,
  \]
  and $x_k = s_k^{-1} \odot g'' \odot s_k$ is also some random walk on graph $\Gamma$.

  Let $G$ be a weighted graph with eigenvalue gap $\epsilon = 1 - \lambda$ and non-uniformity $\nu$. Let random walk on $G$ starts in distribution $q$ and has stationary distribution $\pi$. Then Chernoff bound for expander graphs \citep{Gillman1993} states that for any positive integer $n$ and for any $\gamma > 0$:
  \begin{equation}\label{eq:gillman-chernoff}
  \Pr\left[ \bigg| \sum_{i=1}^n f(x_i) - n \E_\pi f \bigg| \ge \gamma \right] \le 4 N_q \exp\left[-\bigg(\frac{\gamma}{||f||_\infty}\bigg)^2 \frac{\epsilon}{20n} \right].
  \end{equation}

  Here we have graph weights $w_{ij} = \frac 1 d$ for all $i,j$ and $w_x=1$, thus $\nu = 1$ and $\pi(x) = \frac 1 V$. Initial distribution $q$ is uniform distribution over $G$, therefore $N_q = 1$. Function $f(x) = \chi(x)$ obviously has $||f||_\infty \le 1$. We also bound (\ref{eq:gillman-chernoff}) with some small probability, e.g. $\frac 1 {|G|}$. Then (\ref{eq:gillman-chernoff}) becomes
  \[
  \Pr \bigg[ \big| \sum_{i=1}^t f(x_i) - t \E_\pi f \big| \ge \gamma \bigg] \le 4\exp \bigg[ - \frac {\gamma^2 \epsilon} {20t} \bigg] \le \frac 1 {|G|}.
  \]

  Solving with respect to $t$ gives us:
  \[
  t \ge \frac {20}{(1-\lambda) \delta} \ln(4 |G|) = O(\log |G|).
  \]

  If we make a random walk of length $t = O(\log |G|)$, we will get a quantum hash function with high probability.
\end{proof}

So, construction of this quantum hash function requires only $O(\log |G|)$ bits of randomness if underlying expander graph is chosen carefully.

\begin{cor}
  For all $n$ and $\delta \in (0;\frac 1 2)$ there exist a $\delta$-resistant $(\log n;\log t + 1)$ quantum hash function with $t \ge \frac{160 \sqrt 2}{3 \delta} \ln(4 n)$.
\end{cor}

\begin{proof}
  We use Margulis construction \cite{Hoory2006} of $(8;\frac{5 \sqrt 2}{8})$ expander graph with $n^2$ vertices and character of group $\ZZ_n^2$.
\end{proof}

\section{Extractors for Quantum Hashing}\label{sec:extractor-qhf}

\begin{defn}
  Let $\Ext: G \times \Xzo^d \to H$ be a $(k;\epsilon)$ extractor function. Let $t$ and $s_i \in G, i \in \{1,\ldots,t\}$ be parameters. We choose them in Theorem \ref{thm:extractor-qhf}. Denote $S = \{ s_i \}$.

  We define a quantum hash function $\EE$ based on extractor $\Ext$ as follows. 

  \[
  \ExtF(g) = \sum_{i=1}^t \sum_{j=1}^{2^d} \chi(\Ext(g \circ s_i,j)) \ket{j} \ket{i}.
  \]
\end{defn}

Intuitively, we start from several vertices and move along all incident edges simultaneously.

Parameters $t$, $s_i$ depend on security parameter $\epsilon$. Let us choose it.

\begin{thm}\label{thm:extractor-qhf}
  If $\Ext$ is a $(k,\epsilon)$ extractor, parameter $t > \frac{\log |H| + 1}{2\epsilon^2} ||\chi||_\infty$ and $s_i$ are chosen according to distribution $X$ with $H_\infty(X) \ge k$, then $\EE_{\Ext}$ is an $\epsilon$-resistant $(n;d + \log t)$ quantum hash function. 
\end{thm}

\begin{proof}
  It is sufficient to prove that for any $g' \neq g$
  \begin{align*}
    \bigg| \braket{\ExtF(g)}{\ExtF(g')} \bigg| &= \bigg| \sum_{i=1}^t \sum_{j=1}^{2^d} \chi(\Ext(g \circ s_i, j)^{-1} \bullet \Ext(g' \circ s_i, j) ) \bigg| \le \\
    & \le \sum_{i=1}^t \sum_{j=1}^{2^d} | \chi(\Ext(g \circ s_i, j)^{-1} \bullet \Ext(g' \circ s_i, j) ) | < \epsilon.
  \end{align*}
  
  Define $X_i$ to be a distribution of (random variable) $s_i$. Let $Y_i$ be a random variable $\E_{U_d}[|\chi(\Ext(X_i,U_d))|]$.

  It is easy to see that $Y_i \le ||\chi||_\infty = 1$.

  Then by Hoeffding's inequality:
  \[
  Pr \Bigg[ \bigg| \frac 1 t \sum_{i=1}^t Y_i - \E \Big[ \frac 1 t \sum_{i=1}^t Y_i \Big] \bigg| \ge \epsilon \Bigg] \le 2 \exp \bigg( - 2t\epsilon^2 \bigg).
  \]

  Bounding this probability by $\frac 1 {|H|}$ and solving with respect to $t$ gives
  \[
  t \ge \frac{\log|H| + 1}{2\epsilon^2}.
  \]
\end{proof}

Note that selecting parameters $S$ requires $O(\log |G| \times \log |H|)$ random bits.

\begin{cor}
  For every $\epsilon > 0$, $\alpha > 0$ and all positive integers $n, k$ there exist an $\epsilon$-resistant $(n; \log t + d + 1)$ quantum hash function, where $t \ge \frac{m+1}{2\epsilon^2}$, $d = O(\log n + \log (1 / \epsilon))$ and $m \ge (1-\alpha) k$.
\end{cor}

\begin{proof}
  \citet{Guruswami2009} proved that for every $\alpha > 0$ and all positive integers $n,k$ and all $\epsilon > 0$ there is an explicit construction of a $(k;\epsilon)$ extractor $E: \Xzo^n \times \Xzo^d \to \Xzo^m$ with $d  = O(\log n + \log (1 / \epsilon))$ and $m \ge (1-\alpha) k$.

  Quantum hash function $\EE_{E,t}$ is the required function.
\end{proof}




\section{Keyed quantum hash functions}\label{sec:keyed-qhf}

Classical message authentication codes (MAC) have wide range of applications. They are defined as a triple of algorithms: $G$ that generates a key, $S$ that uses the key and the message to generate a tag of the message, and $V$ that uses the key, the message and the tag to verify message integrity.

Formally, $G : 1^n \to K$, where $n$ is a a security parameter and $K$ is a set of all possible keys, $S: K \times X \to T$, where $X$ is a set of messages and $T$ is a set of tags and $V: K \times X \times T \to \{\mathrm{Acc}, \mathrm{Rej}\}$.

We require the following property for MAC to be a sound system:
\begin{equation}\label{eq:mac-sound}
  \forall n, \forall x \in X: k = G(1^n), V \big( k, x, S(k, x) \big) = \mathrm{Acc},
\end{equation}
i.e. that verifier always accepts a generated tag.

We also require that MAC is a secure system and for any adversary $A$ that can query MAC:
\begin{equation}\label{eq:mac-secure}
  \forall n, k \notin \mathrm{Query}(A), (x,t) \gets A(S), \Pr \big[ V(k, x, t) = \mathrm{Acc} \big] \le \mathrm{negl}(n),
\end{equation}
i.e. any adversary that can query MAC outputs correct tag for some key that was not queried and some message with negligible probability.

One classical construction of MAC is hash-based MAC (also known as keyed hash functions). Basically, keyed hash function is a function $H(k,x)$, such that $H(k,\cdot)$ is a cryptographic hash function for every $k$. It is easy to see that such function can be used as MAC.

With the same considerations as in Section \ref{sec:qhf}, we define these algorithms to be the following.

\begin{defn}
  An $(\epsilon,\delta)$ keyed quantum hash function is a quantum function $S$, such that

  A function $S$ accepts a key $k \in K$ and a message $x \in X$ and outputs a quantum tag for $x$: $S: K \times X \to T = \HH{t}$.

  We require soundness, i.e. tags should be different for different messages under the same key.
  \[
  \forall k \in K, \forall x \in X, \forall y \neq x: \braket{S(k,x)}{S(k,y)} < \epsilon.
  \]
  For $x = y$ we get $\braket{S(k,x)}{S(k,x)} = 1$.

  We also require unforgeability:
  \[
  \forall k \in K, k \notin \mathrm{Query}(A), (x,t) \gets A(S), \Pr \big[ \braket{t}{S(k,x)} \ge \epsilon) \big] \le \delta,
  \]
  where $A$ is arbitrary attacker that can query $S$ and $\mathrm(Query)(A)$ is a set of queries made.
\end{defn}

Informally, keyed quantum hash function outputs a tag for a message. If someone changes a message, then the verification step fails with high probability. If an attacker Eve can query a keyed quantum hash function, access to a function doesn't help her to forge a tag for some message with some (unqueried) key.

\begin{thm}\label{thm:extractor-keyed-qhf}
  Let us define an extractor-based keyed quantum hash function as follows. Let $\Ext: \Xzo^n \times \Xzo^d \to \Xzo^m$ be a $(k,b,\epsilon)$ extractor against $b$ quantum storage and $b > r (d + \log t)$ .

  Then a function
  \[
  \ExtKF(\key,g) = \sum_{i=1}^t \sum_{j=1}^{2^d} \chi(\Ext(g \circ \key \circ s_i,j)) \ket{j}\ket{i}
  \]
  is a $(\epsilon; \epsilon + \epsilon^{2^s+1})$ keyed quantum hash function secure against an attacker $A$ with access to $r$ queries to $\EE_{\Ext}$.
\end{thm}

\begin{proof}
  We have to prove two claims. First, for any $k,x$ and $x' \neq x$, it holds that $\braket{\ExtKF(k,x)}{\ExtKF(k,x')} < \epsilon$. Second, for any attacker $A$ and any $k \notin \mathrm{Query}(A)$ attacker output $x,t$ such that $\braket{t}{\ExtKF(k,x)} \ge \epsilon$  with negligible probability.

  The first claim is implied by Theorem \ref{thm:extractor-qhf}.

  To prove the second claim we note that access to hash function doesn't help attacker to output correct tag. Proof by contradiction. Suppose $A$ to be such attacker. Then we can distinguish between $\Ext(X,U_d)$ and $U_m$ using a $r (\log t + d)$ qubits. But $r (\log t + d) < b$ that contradicts the fact that $\Ext$ is an extractor against $b$ quantum storage.

  Then attacker should output the tag without access to hash function. This is equal to outputting a state that is close to correct tag. Then the probability of correct guessing $p$ is a ratio of the volume of sphere with radius $\epsilon$ to the volume of the whole space:
  \[
  p = \frac {c \epsilon^{2^s + 1}} {c (1 + \epsilon)^{2^s+1}} \le \epsilon^{2^s+1}.
  \]
\end{proof}

\begin{cor}
  For all positive integers $k,n$ and all $c > 0$ there exist a $(N^{-c}; 2 N^{-c})$ keyed quantum hash function.
\end{cor}

\begin{proof}
  \citet{De2009} proved that for every $\alpha, c > 0$ there exist an explicit $(\alpha N,b,N^{-c})$ extractor $E: \Xzo^N \times \Xzo^d \to \Xzo^m$ against $b$ quantum storage with $d = O(\log^4 n)$ and $m = \Omega (\alpha N - b)$.
\end{proof}






\section{Open problems}

Groups that we considered here and all constructions known to us use finite groups or sets and hash input strings of finite lengths.

\begin{question}
  Can quantum hash functions be constructed for infinite groups?
\end{question}

On the one hand, even one qubit can store arbitrary length binary string. On the other hand, the measurement of one qubit can't result in more than one classical bit of information.

And ``dual'' question about infinite strings.

\begin{question}
  Can quantum hash functions work on infinite input strings (i.e. $\Xzo^*$)?
\end{question}

This problem seems to be easier, but it probably requires careful analysis.

Another interesting line of research would be improving keyed quantum hash function.

\begin{question}
  Can keyed quantum hash function be secure against an attacker with unlimited number of queries?
\end{question}



\end{document}